\documentclass{llncs}
\usepackage{times}
\usepackage{amsmath}
\usepackage{amssymb}
\usepackage{algorithm}
\usepackage[noend]{algorithmic}
\usepackage{color, array, colortbl}
\usepackage{graphicx}
\usepackage{multirow}
\usepackage{wrapfig}
\usepackage{comment}
\usepackage{tikz}
\usepackage{soul}

\usepackage{tabularx}
\usepackage{mdwlist}
\usepackage{xspace}
\usepackage{booktabs}
\usepackage[caption=false,font=footnotesize]{subfig}
\usepackage{framed}
\usepackage{color}
\definecolor{shadecolor}{rgb}{0.9,0.9,0.9}

\newcommand{\COMMENTED}[1]{}

\newcommand{\cancel}[1]{}

\newcommand{\stefan}[1]{}

\newcommand{\Bf}{\textsc{Bf}}
\newcommand{\Af}{\textsc{Af}}

\newcommand{\ALG}{\textsc{Alg}}
\newcommand{\TREE}{\textsc{Tree}}

\newcommand{\FRAME}{\textsc{Dict}}

\newcommand{\pirev}[1]{\ensuremath{\pi^{-1}(#1)}}

\begin{document}

\title{Request Complexity of VNet Topology Extraction:\\Dictionary-Based Attacks\thanks{This project was partly funded by the Secured Virtual Cloud (SVC) project.}}

\author {
   Yvonne-Anne Pignolet$^1$, Stefan Schmid$^2$, Gilles Tredan$^3$\\
   \small $^1$ ABB Corporate Research, Switzerland\\
   \small $^2$ Telekom Innovation Laboratories \& TU Berlin, Germany\\
   \small $^3$ CNRS-LAAS, France\ }

\institute{}

\date{}

\maketitle \thispagestyle{empty}

\sloppy

\begin{abstract}
The network virtualization paradigm envisions an Internet where arbitrary
virtual networks (VNets) can be specified and embedded over a shared substrate
(e.g., the physical infrastructure). As VNets can be requested at short notice
and for a desired time period only, the paradigm enables a flexible service
deployment and an efficient resource utilization.

This paper investigates the security implications of such an architecture. We
consider a simple model where an attacker seeks to extract secret
information about the substrate topology, by issuing repeated VNet embedding
requests. We present a general framework that exploits basic properties of the VNet embedding relation to infer the entire topology. Our framework is based on a graph motif dictionary applicable for various graph classes. We provide upper bounds on the \emph{request complexity}, the number of requests needed by the attacker to succeed. Moreover, we present some experiments on existing networks to evaluate this dictionary-based approach.
\end{abstract}

\section{Introduction}

While network virtualization enables a flexible resource sharing, opening the infrastructure
for automated virtual network (VNet) embeddings or service deployments may introduce new
kinds of security threats. For example, by virtualizing its network infrastructure (e.g.,
the links in the aggregation or backbone network, or the computational or storage resources
at the points-of-presence), an Internet Service Provider (ISP) may lose control over how its network
is used. Even if the ISP manages the allocation and migration of VNet slices and services
itself and only provides a very rudimentary interface to interact with customers (e.g., service or
content providers), an attacker may infer information about the network topology (and state) by generating VNet
requests.

This paper builds upon the model introduced in~\cite{minicom13} and studies complexity of the
\emph{topology extraction problem}: How many VNet requests are required to infer the full
topology of the infrastructure network? While algorithms for trees and cactus graphs with request complexity $O(n)$ and a lower bound for general graphs of $\Omega(n^2)$ have been shown in~\cite{minicom13}, graph classes between these extremes have not been studied.

\textbf{Contribution.}
This paper presents a general framework to solve the topology extraction problem.
We first describe necessary and sufficient conditions which facilitate the ``greedy''
exploration of the substrate topology (the \emph{host graph} $H$) by iteratively extending the
requested VNet graph (the \emph{guest graph} $G$). Our framework then exploits these conditions to
construct an ordered (request) \emph{dictionary} defined over so-called \emph{graph motifs}.
We show how to apply the framework to different graph families, discuss the implications on the request complexity,
and also report on a small simulation study on realistic topologies. These empirical results show that many scenarios
can indeed be captured with a small dictionary, and small motifs are sufficient to infer if not the entire, then
 at least a significant fraction of the topology.


\section{Background}\label{sec:model}

This section presents our model and discusses how it compares to related work.

\textbf{Model.}
The VNet embedding based topology extraction problem has been introduced
in~\cite{minicom13}. The formal setting consists of two entities: a
\emph{customer} (the ``adversary'') that issues virtual network (VNet)
requests and a \emph{provider} that performs the access control and the
embedding of VNets. We model the virtual network requests as simple,
undirected graphs $G=(V,E)$ (the \emph{guest graph}) where $V$ denotes the
virtual nodes and $E$ denotes the virtual edges connecting nodes in $V$.
Similarly, the infrastructure network is given as an undirected graph
$H=(V,E)$ (the so-called \emph{host graph} or \emph{substrate}) as well, where
$V$ denotes the set of substrate nodes, $E$ is the set of substrate links, and
$w$ is a capacity function describing the available resources on a given node
or edge. Without loss of generality, we assume that $H$ is connected and that
there are no parallel edges or self-loops neither in VNet requests nor in the
substrate.

In this paper we assume that besides the resource demands, the VNet
requests do not impose any mapping restrictions, i.e., a virtual
node can be mapped to \emph{any} substrate node, and we assume that a virtual link
connecting two substrate nodes can be mapped to an entire (but single)
\emph{path} on the substrate as long as the demanded capacity is available. These assumptions are typical
for virtual networks~\cite{virsurvey}.

A virtual link which is mapped to more than one substrate link
however can entail certain costs at the \emph{relay nodes}, the substrate
nodes which do not constitute endpoints of the virtual link and merely serve
for forwarding.
We model these costs with a parameter $\epsilon> 0$ (per link). Moreover, we also allow multiple virtual nodes to be mapped to the same
substrate node if the node capacity allows it; we assume that if two virtual nodes are mapped
to the same substrate node, the cost of a virtual link between them is zero.
\begin{definition}[Embedding $\pi$, Relation $\mapsto$]\label{def:embedding}
An \emph{embedding} of a graph $A=(V_A,E_A,w_A)$ to a graph $B=(V_B,E_B,w_B)$ is a mapping
$\pi: A \to B$ where every node of $A$ is mapped to exactly one node of $B$, and every edge of $A$ is mapped to
a path of $B$. That is, $\pi$ consists of a node $\pi_V : V_A \to V_B$ and
an edge mapping $\pi_E : E_A \to P_B$, where $P_B$ denotes the set of paths.
We will refer to the set of virtual nodes embedded on a node $v_B\in V_B$ by
$\pi^{-1}_V(v_B)$; similarly, $\pi^{-1}_E(e_B)$ describes the set of virtual
links passing through $e_B\in E_B$ and $\pi^{-1}_E(v_B)$ describes the virtual links passing through $v_B\in V_B$ with $v_B$ serving as a relay
node.

To be valid, the embedding $\pi$ has to fulfill the following properties: ($i$)
Each node $v_A\in V_A$ is mapped to exactly one node $v_B\in V_B$ (but given
sufficient capacities, $v_B$ can host multiple nodes from $V_A$). ($ii$) Links
are mapped consistently, i.e., for two nodes $v_A,v_A'\in V_A$, if
$e_A=\{v_A,v_A' \}\in E_A$ then $e_A$ is mapped to a single (possibly empty
and undirected) path in $B$ connecting nodes $\pi(v_A)$ and $\pi(v_A')$. A link $e_A$ cannot be split into multiple paths. ($iii$) The capacities
of substrate nodes are not exceeded: $\forall v_B \in V_B$: $\sum_{u\in
\pi_V^{-1}(v_B)} w(u) +\epsilon \cdot |\pi_E^{-1}(v_B)| \leq w(v_B)$.
($iv$) The capacities in $E_B$ are respected as well, i.e., $\forall e_B \in E_B$: $\sum_{e\in \pi^{-1}_E(e_B)} w(e) \leq w(e_B)$.

If there exists such a valid embedding mapping $\pi$, we say that graph $A$ can be embedded in $B$,
denoted by $A\mapsto B$. Hence, $\mapsto$ denotes the VNet \emph{embedding relation}.
\end{definition}

The provider has a flexible choice where to embed a VNet as long as a valid mapping is
chosen. In order to design topology discovery algorithms, we exploit the following property of the embedding relation.
\begin{lemma}\label{lemma:poset}
The embedding relation $\mapsto$ applied to any family $\mathcal{G}$ of undirected graphs (short: $(\mathcal{G},\mapsto)$), forms a partially ordered set (a \emph{poset}). [Proof in Appendix]
\end{lemma}

We are interested in algorithms that ``guess'' the target
topology $H$ (the host graph) among the set $\mathcal{H}$ of possible substrate topologies. Concretely, we assume that given a VNet request $G$ (a guest graph), the substrate provider
always responds with \emph{an honest (binary) reply} $R$ informing the customer whether the
requested VNet $G$ is embeddedable on the substrate $H$.
Based on this reply, the attacker may then decide to ask the provider to embed the corresponding
VNet $G$ on $H$, or it may not embed it and continue asking for other VNets.
Let $\ALG$ be an algorithm asking a series of requests
$G_1,\ldots, G_t$ to reveal $H$.
The \emph{request complexity} to infer the topology 
is measured in the number of requests $t$ (in the worst case) until $\ALG$ issues a request $G_t$ which is isomorphic to $H$ and \emph{terminates}
(i.e., $\ALG$ knows that $H=G_t$ and does not issue  further requests).\\

\textbf{Related Work.}
Embedding VNets is an intensively studied problem and there exists a large body of literature (e.g.,~\cite{ammar,pb-embed,ucc12mip,infocom12}), also on distributed computing approaches~\cite{distmapping} and online algorithms~\cite{podc11,moti}.
Our work is orthogonal to this line of literature in the sense that
we assume that an (arbitrary and not necessarily resource-optimal) embedding algorithm is \emph{given}.
Instead, we focus on the question of how the feedback obtained
through these algorithms can be exploited, and we study the
implications on the information which can be obtained about a provider's infrastructure.

Our work studies a new kind of topology inference problem. Traditionally, much graph discovery research has been conducted in the context of today's complex
networks such as the Internet which have
fascinated scientists for many years, and there exists a wealth of results on the topic.
The
classic instrument to discover Internet topologies is
\emph{traceroute}~\cite{traceroutedata}, but the tool has several problems which makes the problem challenging. One complication of traceroute stems from the
fact that routers may appear as stars (i.e., anonymous nodes), which renders the accurate characterization of Internet
topologies difficult~\cite{icdcn11,disc11,infocom03}.
\emph{Network tomography} is another important field of topology discovery. In network tomography, topologies are explored using pairwise end-to-end
measurements, without the cooperation of nodes along these paths.
This approach is quite flexible and applicable in various contexts,
e.g., in social networks. For a good discussion of this approach as
well as results for a routing model along shortest and second
shortest paths see~\cite{sigmetrics11}. For
example,~\cite{sigmetrics11} shows that for sparse random graphs, a
relatively small number of cooperating participants is sufficient to
discover a network fairly well.
Both the traceroute and the network tomography problems differ from our virtual network topology discovery problem in that the exploration
there is inherently \emph{path-based} while we can ask for entire virtual graphs.

The paper closest to ours is~\cite{minicom13}. It introduces the topology extraction model studied in this paper, and presents
an asymptotically optimal algorithm for the cactus graph family (request complexity $\Theta(n)$), as well as a general algorithm
(based on spanning trees) with request complexity $\Theta(n^2)$.

\section{Motif-Based Dictionary Framework}\label{sec:framework-app}

The algorithms for tree and cactus graphs presented in~\cite{minicom13}
can be extended to a framework for the discovery of more general graph
classes. It is based on the idea of growing sequences of subgraphs from nodes
discovered so far. Intuitively, in order to describe the ``knitting'' of
a given part of a graph, it is often sufficient to use a small set of graph
\emph{motifs}, without specifying all the details of how many substrate nodes
are required to realize the motif. We start this section with the introduction of motifs and their composition and expansion. Then we present the dictionary concept, which structures motif sequences in a way that enables the efficient host graph discovery with algorithm~$\FRAME$. Subsequently, we give some examples and finally provide the formal analysis of the request complexity.

\subsection{Motifs: Composition and Expansion}

In order to define the motif set of a graph family $\mathcal{H}$, we need the concept of \emph{chain (graph) $C$}: $C$ is just a graph $G=(\{v_1,v_2\}, \{v_1,v_2\})$
consisting of two nodes and a single link. As its edge represents a virtual link that may be embedded along entire path in the substrate network, it is called a \emph{chain}.

\begin{definition}[Motif\label{def:motif-app}]
Given a graph family $\mathcal{H}$, the set of motifs of $\mathcal{H}$ is defined constructively: If any member of $H\in\mathcal{H}$
has an edge cut of size one, the \emph{chain $C$} is a motif for $\mathcal{H}$.
All remaining motifs are at least 2-connected (i.e.,
any pair of nodes in a motif is connected by at least two vertex-disjoint paths). These motifs
can be derived by the at least 2-connected components of any $H\in\mathcal{H}$ by repeatedly removing all nodes with degree smaller
or equal than two from $H$ (such nodes do not contribute to the knitting) and merging the incident edges, as long as all remaining cycles do not contain parallel edges.
Only one instance
of isomorphic motifs is kept.
\end{definition}

Note that the set of motifs of $\mathcal{H}$ can also be computed by iteratively by removing all low-degree nodes and subsequently determine the graphs connecting nodes constituting a vertex-cut of size one for each member $H\in \mathcal{H}$.
In other words, the motif set $\mathcal{M}$ of a graph family
$\mathcal{H}$ is a set of non-isomorphic minimal (in terms of number of nodes) graphs that are required to
construct each member $H\in\mathcal{H}$ by taking a motif and either replacing edges with two edges connected by a node or gluing together components several times. More formally, a graph family containing all elements of $\mathcal{H}$ can be constructed by applying the following rules repeatedly.
\begin{definition}[Rules]\label{def:motif_rules}
(1) Create a new graph consisting of a motif $M\in\mathcal{M}$ (\emph{New Motif Rule}).
(2) Given a graph created by these rules, replace an edge $e$ of $H$ by a new node and two new edges connecting the incident nodes of $e$ to the new node
(\emph{Insert Node Rule}).
(3) Given two graphs created by these rules, attach them to each other such that they share exactly one node (\emph{Merge Rule}).
\end{definition}

Being the inverse operations of the ones to determine the motif set, these rules are sufficient to compose all graphs in $\mathcal{H}$:
If $\mathcal{M}$  includes all motifs of $\mathcal{H}$, it also includes all 2-connected components of $H$,
according to Definition~\ref{def:motif-app}. These motifs can be glued together using the \emph{Merge Rule}, and
eventually the low-degree nodes can be added using the \emph{Insert Node Rule}.
Therefore, we have the following lemma.
\begin{lemma}\label{lemma:equiv}
Given the motifs $\mathcal{M}$ of a graph family $\mathcal{H}$, the repeated application of the rules in Definition~\ref{def:motif_rules} allows us to construct each member $H\in\mathcal{H}$.
\end{lemma}

However, note that it may also be possible to use these rules to construct graphs that are \emph{not} part of the family.
The following lemma shows that when degree-two nodes are added to a motif $M$ to form a graph $G$, all network elements (substrate nodes and links)
are \emph{used} when embedding $M$ in $G$ (i.e., $M\mapsto G$).
\begin{lemma}\label{lem:nodesLoad-app}
  Let $M \in (\mathcal{M}\setminus\{C\})$ be an arbitrary two-connected motif,
  and let $G$ be a graph obtained by
  applying the \emph{Insert Node Rule} (Rule $2$ of Definition~\ref{def:motif_rules}) to motif $M$.
  Then, an embedding $M\mapsto G$ involves all nodes and edges in $G$: at least $\epsilon$ resources are used on all nodes and edges.
\end{lemma}
\begin{proof}  Let $v\in G$. Clearly, if there exists $ u \in M$ such that $v= \pi(u)$,
  then $v$'s capacity is used fully.
  Otherwise, $v$ was added by Rule~$2$. Let $a,b$ be the two nodes of $G$ between which Rule~$2$ was applied, and hence
  $\{\pirev{a},\pirev{b}\} \in E_M$ must be a motif edge.
  Observe that for these nodes' degrees it holds that $\deg(a)=\deg(\pirev{a})$ and $\deg(b)=\deg(\pirev{b})$ since Rule~$2$ never modifies the
  degree of the old nodes in the host graph $G$. Since links are of unit capacity, each
  substrate link can only be used once: at $a$ at most $\deg(a)$
  edge-disjoint paths can originate, which yields a contradiction to the
  degree bound, and the relaying node $v$ has a load of $\epsilon$.
\hfill$\Box$ \end{proof}

Lemma~\ref{lem:nodesLoad-app} implies that no additional
nodes can be inserted to an existing embedding. In other words, a motif constitutes a ``minimal reservation pattern''. As we will see, our algorithm
will exploit this invariant that motifs cover the entire graph knitting, and adds simple nodes (of degree 2)
only in a later phase.
\begin{corollary}\label{cor:embedItself-app}
  Let $M \in (\mathcal{M} \backslash
  \{C\})$ and let $G$ be a graph obtained by applying Rule~$2$ of Definition~\ref{def:motif_rules} to motif $M$.
  Then, no additional node can be embedded on $G$ after embedding $M\mapsto G$.
\end{corollary}

Next, we want to \emph{combine} motifs explore larger ``knittings'' of graphs. Each motif pair is glued together at a single node \emph{or edge} (``attachment point''):
We need to be able to conceptually join to motifs at edges as well because the corresponding edge of the motif can be expanded by the \emph{Insert Node Rule} to create a node where the motifs can be joined.

\begin{definition}[Motif Sequences, Subsequences, Attachment Points, $\prec$]
A \emph{motif sequence} $S$ is a list $S=(M_1{a_1a'_1}M_2\ldots M_k)$ where $\forall i:~M_i\in
\mathcal{M}$ and where $M_i$ is glued together at exactly one node with $M_{i-1}$ (i.e., $M_i$ is ``attached'' to a node of motif $M_{i-1}$): the notation $M_{i-1}{a_{i-1}a'_{i-1}}M_{i}$ specifies
the selected attachment points $a_{i-1}$ and $a'_{i-1}$.
If the attachment points are irrelevant, we use the notation $S=(M_1M_2\ldots M_k)$ and $M_i^k$ denotes an arbitrary sequence consisting of $k$ instances of $M_i$.
If $S$ can be decomposed into $S=S_1S_2S_3$, where $S_1, S_2$ and $S_3$ are (possibly empty) motif sequences as well, then $S_1, S_2$
and $S_3$ are called \emph{subsequences} of $S$, denoted by~$\prec$.
\end{definition}

In the following, we will sometimes use the \emph{Kleene star} notation $X^{\star}$ to denote a sequence of (zero or more) elements of $X$ attached to each other.

\begin{figure*}
\begin{center}
\includegraphics[width=0.88\columnwidth]{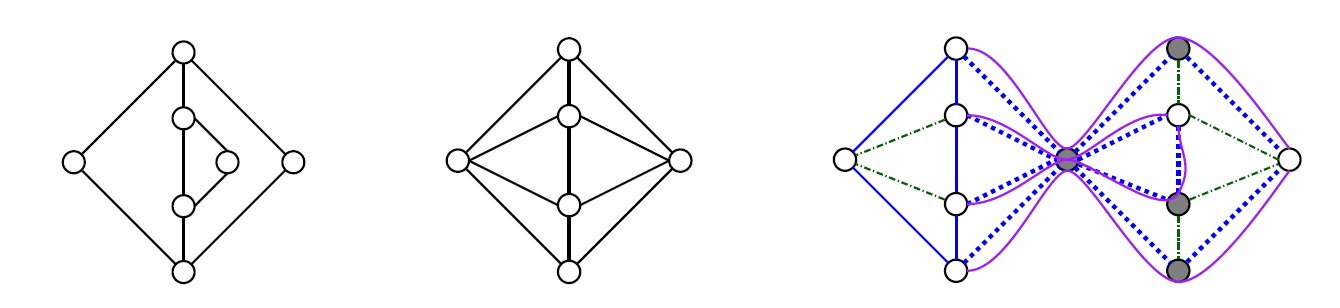}
\vspace{-1.0\baselineskip}
\end{center}
\caption{\emph{Left:} Motif $A$. \emph{Center:} Motif $B$. Observe that $A$
  $\not \mapsto$ $B$. \emph{Right:} Motif $A$ is embedded into two consecutive Motifs $B$:
solid lines are virtual
  links mapped on single substrate links, solid curves are virtual links mapped
  on multiple substrate links, dotted lines are substrate links implementing a
  multi-hop virtual link, and dashed lines are substrate unused links. Grayed
  nodes are relay-only nodes. Observe that the central node has a relaying load
  of $4\epsilon$.}
\label{fig:counterExAB-app}
\vspace{-1.0\baselineskip}
\end{figure*}

One has to be careful when arguing about the embedding of motif sequences, as illustrated in Figure~\ref{fig:counterExAB-app} which shows a counter example for
$M_i \not \mapsto M_j \Rightarrow \forall
k>0, M_i \not \mapsto M_j^k$. This means that we typically cannot just incrementally add motif occurrences to discover a certain substructure.
This is the motivation for introducing the concept of a \emph{dictionary}
which imposes an order on motif sequences and their attachment points.

\subsection{Dictionary Structure and Existence}

 In a nutshell, a dictionary is a \emph{Directed Acyclic Graph (DAG)} defined over all possible motifs
  $\mathcal{M}$.  and imposes an order (poset relationship $\mapsto$) on problematic motif sequences which need to be
 embedded one before the other (e.g., the composition
  depicted in Figure \ref{fig:counterExAB-app}). To distinguish them from sequences, dictionary entries are called \emph{words}.

\begin{definition}[Dictionary, Words]
  A \emph{dictionary} $D(V_D,E_D)$ is a directed acyclic graph (DAG) over a set of motif sequences $V_D$ together
  with their attachment points. In the context of the dictionary, we will call a motif sequence \emph{word}.
  The links $E_D$ represent the poset
 embedding relationship $\mapsto$.

 Concretely, the DAG has
  a single root $r$, namely the chain graph $C$ (with two attachment points). In general, the \emph{attachment points}
  of each vertex $v\in V_D$ describing a word $w$ define how $w$ can be connected to other words.  The directed
  edges $E_D=(v_1,v_2)$ represent the transitively reduced embedding poset
  relation with the chain $C$ context: $C v_1 C$ is embeddable in $C v_2 C$ and there is
  no other word $C v_3 C$ such that $C v_1 C\mapsto C v_3 C$, $C v_3
    C\mapsto C v_2 C$ and $C v_3 C\not\mapsto C v_1 C$ holds. (The
    chains before and after the words are added to ensure that attachment points
    are ``used'': there is no edge between two isomorphic words with different
    attachment point pairs.)

 We require that the dictionary be \emph{robust to composition}: For any
  node $v$, let $R_v= \{ v' \in V_D, v\mapsto v'\}$ denote the ``reachable'' set
  of words in the graph and $\overline{R}_v = V_D \setminus
  R_i$ all other words. We require that $v \not \mapsto W, \forall W \in
  Q_i :=\overline{R}_i^\star \backslash R_i^\star$, where the transitive
  closure operator $X^\star$ denotes an arbitrary sequence (including the empty
  sequence) of elements in $X$ (according to their attachment points).

\end{definition}

See Figure~\ref{fig:dictionary} for an example.
Informally, the robustness requirement means that the word represented by $v$ cannot be embedded in
any sequence of ``smaller'' words, unless a subsequence of this sequence is in
the dictionary as well. As an example, in a dictionary containing motifs
$A$ and $B$ from Figure~\ref{fig:counterExAB-app} would contain vertices $A$,
$B$ and also $BB$, and a path from $A$ to $BB$.
\begin{figure}[h!]
	\centering
		a)\includegraphics[width=0.480\columnwidth]{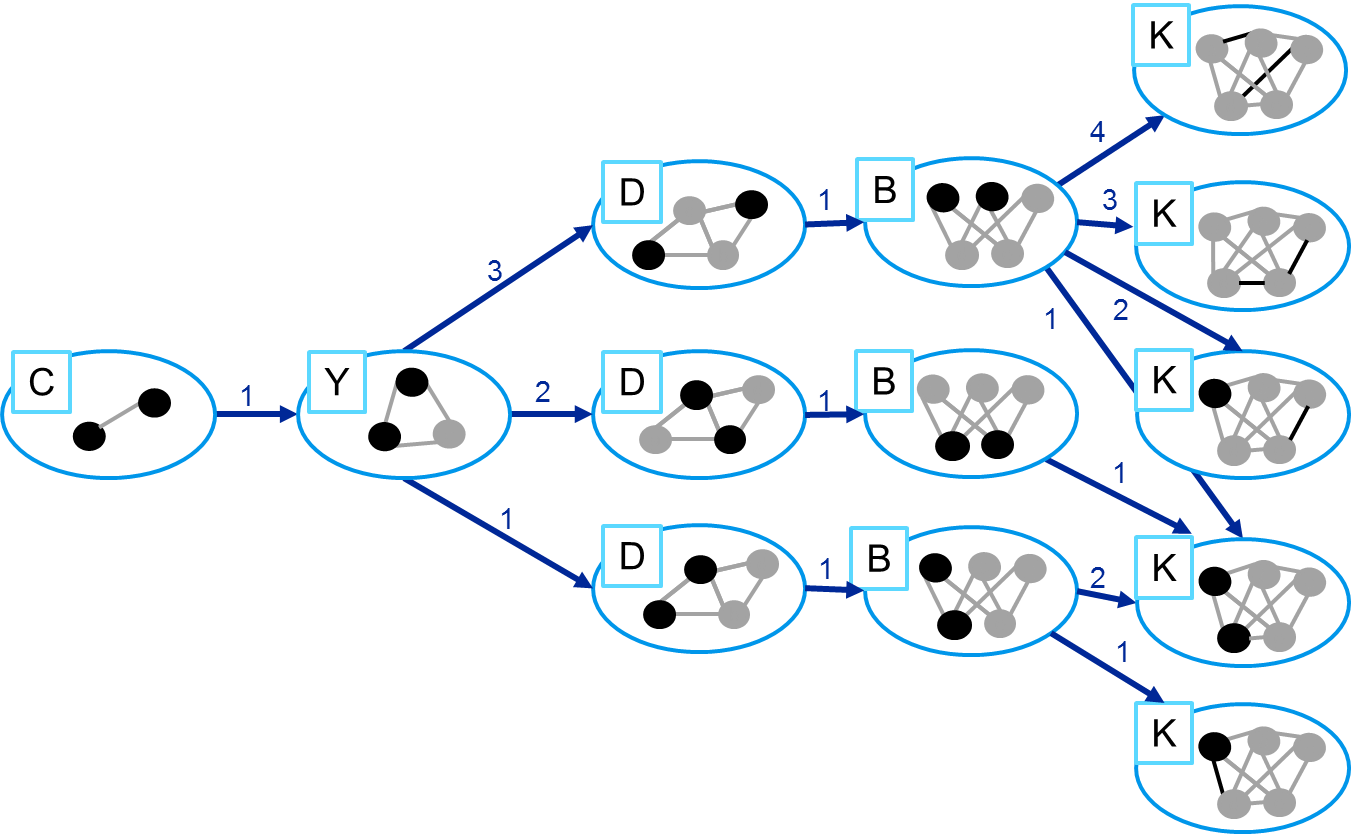}~~~~
		b) \includegraphics[width=0.36\columnwidth]{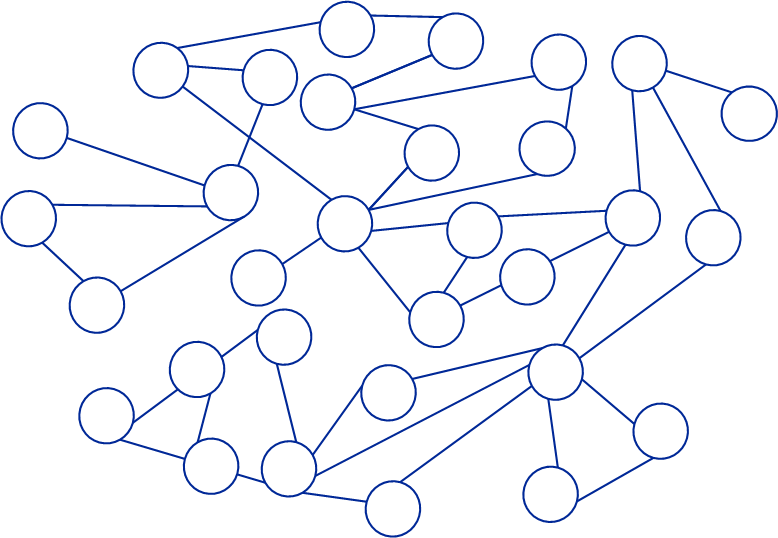}
	\vspace{-.5\baselineskip}\caption{a) Example dictionary with motifs Chain $C$, Cycle $Y$, Diamond $D$, complete bipartite graph $B = K_{2,3}$ and complete graph $K = K_5$. The attachment point pair of each word is black, the other nodes and edges of the words are grey. The edges of the dictionary are locally labeled, which is used in $\FRAME$ later. b) A graph that can be constructed from the dictionary words.}
	\label{fig:dictionary}
\vspace{-1.0\baselineskip}
\end{figure}
  In the following, we use the notation $\max_{v\in V_D}(v\mapsto S)$ to
  denote the set of ``maximal'' vertices with respect to their embeddability
  into $S$: $ i \in \max_{v\in V_D}(v\mapsto S) \Leftrightarrow (i \mapsto S)
  \wedge (\forall j \in \Gamma^+(i), j \not \mapsto S)$, where $\Gamma^+(v)$ denotes the set of outgoing neighbors of $v$.
Furthermore, we say that a dictionary $D$ \emph{covers} a motif sequence
$S$ iff $S$ can be formed by concatenating dictionary words (henceforth denoted
by $S \in D^\star$) at the specified attachment points. More generally, a
dictionary covers a graph, if it can be formed by merging sequences of
$D^\star$.

Let us now derive some properties of the dictionary which are crucial for a
proper substrate topology discovery. First we consider maximal dictionary words
which can serve as embedding ``anchors'' in our algorithm.
\begin{lemma}
\label{lemma:dico-property-app}
Let $D$ be a dictionary covering a sequence $S$ of motifs, and let $i \in
\max_{v\in V_D}(v \mapsto S)$. Then $i$ constitutes a subsequence of $S$, i.e.,
$S$ can be decomposed to $S_1iS_2$, and $S$ contains no words of order
at most $i$, i.e., $S_1,S_2 \in (\overline{R}_i \cup \{i\})^\star$.
\end{lemma}
\begin{proof}
  By contradiction assume $i \in \max_{v\in V_D}(v \mapsto S)$ and $i$ is not a
  subsequence of $S$ (written $i \not\prec S$).  Since $D$ \emph{covers} $S$ we
  have $S \in V_D^\star$ by definition.


  Since $D$ is a dictionary and $i \mapsto S$ we know that $S \not \in
  Q_i$. Thus, $S \in D^\star \backslash Q_i$: $S$ has a
  subsequence of at least one word in $R_i$.  Thus there exists $k \in R_i$ such
  that $k \prec S$. If $k=i$ this implies $i \prec S$ which contradicts our
  assumption. Otherwise it means that $\exists j \in \Gamma^+(i)$ such that
  $j\mapsto k \prec S$, which contradicts the definition of $i\in \max_{v\in
    V_D}(v \mapsto S)$ and thus it must hold that $i\prec S$.
  \hfill$\Box$ \end{proof}

The following corollary is a direct consequence of the definition of $i \in \max_{v\in V_D}(v \mapsto S)$ and
Lemma~\ref{lemma:dico-property-app}: since for a motif sequence $S$ with $S
\in (\overline{R}_i \cup \{i\})^\star$, all the subsequences of $S$ that contain no
$i$ are in $\overline{R}_i^\star$.  As we will see, the corollary is useful to
identify the motif words composing a graph sequence, from the most complex words
to the least complex ones.
\begin{corollary}\label{cor:recursive-app}
  Let $D$ be a dictionary covering a motif sequence $S$, and let $i \in
  \max_{v\in V_D}(v \mapsto S)$. Then $S$ can be decomposed as a sequence $S =
  T_1 i T_2 i ,\ldots, i T_k$ with $T_j \in Q_i, \forall
  j=1,\ldots,k$.
\end{corollary}

This corollary can be applied recursively to describe a motif sequence as a
sequence of dictionary entries. Note that a dictionary always exists.


\begin{lemma}
\label{lem:dicoExistence}
There exists a dictionary $D=(V_D,E_D)$ that covers all member graphs $H$ of a motif graph family $\mathcal{H}$ with $n$ vertices. [Proof in Appendix]
\end{lemma}

\subsection{The Dictionary Algorithm}

With these concepts in mind, we are ready to describe our generalized graph discovery
algorithm called $\FRAME$ (cf~Algorithm~\ref{alg:motifRec-app}).
Basically, $\FRAME$ always grows a request graph $G=H'$ until it
is isomorphic to
$H$ (the graph to be discovered). This graph growing is performed according to the dictionary, i.e.,
we try to embed new motifs in the order imposed by the dictionary DAG.

$\FRAME$ is based on the observation that it is very costly to discover additional edges between nodes in a 2-connected component:
essentially, finding a single such edge requires testing all possibilities, which is quadratic in the component size.
Thus, it is crucial to first explore the basic ``knitting''
of the topology, i.e., the minors which are at least 2-connected (the \emph{motifs}).
In other words, we maintain the invariant that there are never two nodes $u,v$ which are not
$k$-connected in the currently requested graph $H'$ while they are $k$-connected in $H$; no path
relevant for the connectivity is overlooked and needs to be found later.

Nodes and edges which are not contributing to the
connectivity need not be explored at this stage yet, as they can be efficiently added later. Concretely, these additional nodes can then be discovered
by (1) using an \emph{edge expansion} (where
additional degree two nodes are added along a motif edge), and by (2) adding ``chains'' $C$ to the nodes (a virtual link $C$ constitutes an edge cut of size
one and can again be expanded to entire chain of nodes using \emph{edge expansion}).

Let us specify the \emph{topological order} in which algorithm~$\FRAME$
discovers the dictionary words.  First, for each node $v$ in $V_D$, we define an order on its
outgoing edges $\{(v,w) | w \in \Gamma^+(v)\}$. This order is sometimes referred to
as a ``port labeling'', and each path from the dictionary root (the chain $C$) to a node in
$V_D$ can be represented as the sequence of port labels at each traversed node
$(l_1,l_2,\ldots,l_l)$, where $l_1$ corresponds to a port number in $C$. We can simply use the lexicographic order on integers, $<^d$: $(a_1,
a_2, \dots, a_{n_1}) <^d (b_1,b_2, \dots, b_{n_2}) \iff ((\exists\ m > 0) \
(\forall\ i < m) (a_i = b_i) \land (a_m < b_m)) \vee( \forall i \in \{1,\ldots n_1\},
(a_i = b_i) \land (n_1<n_2)) $, to associate each vertex with its
minimal sequence, and sort vertices of $V_D$ according to their embedding
order. Let $r$ be the \emph{rank} function associating each vertex with its position in this
sorting: $r:V_D\to \{1,\ldots |V_D|\}$ (i.e., $r$ is the topological ordering of $D$).

The fact that subsequences can be defined recursively using a dictionary
(Lemma~\ref{lemma:dico-property-app} and Corollary~\ref{cor:recursive-app}) is
exploited by algorithm~$\FRAME$.  Concretely, we apply
Corollary~\ref{cor:recursive-app} to gradually identify the words composing a
graph sequence, from the most complex words to the least complex ones. This is
achieved by traversing the dictionary depth-first, starting from the root $C$ up
to a maximal node: algorithm~$\FRAME$ tests the nodes of $\Gamma^+(v)$ in
increasing port order as defined
above. 
As a shorthand, the word $v \in V_D$ with $r(v)=i$ is written as $D[i]$;
similarly $D[i]<D[j]$ holds if $r(D[i])<r(D[j])$, a notation that will get
useful to translate the fact that $D[j]$ will be detected before $D[i]$ by
algorithm~$\FRAME$. As a consequence, the word of a sequence $S$ that gets
matched first is uniquely identified: it is $i=\arg\max_x(D[x] \mapsto S)=$
$\max\{r(v)|v\in \max_{v'\in V_D}(v' \mapsto S)\}$: $i$ denotes the maximal word
in $S$.


Algorithm~$\FRAME$ distinguishes whether the subsequences next to a word $v\in V_D$ are
empty ($\emptyset$) or chains ($C$), and we will refer to the subsequence before
$v$ by $\Bf$ and to the subsequence after $v$ by $\Af$. Concretely, while
recursively exploring a sequence between two already discovered parts $T_<$
and $T_>$ we check whether the maximal word $v$ is directly next to $T_<$ (i.e.,
$T_<~v,\ldots,~T_>$) or $T_>$ or both ($\emptyset$), or whether $v$ is somewhere
in the middle.  In the latter case, we add a chain ($C$) to be able to find the
greatest possible word in a next step.

$\FRAME$ uses tuples of the form $(i,j,\Bf,\Af)$ where $i,j \in \mathbb{N}^2$
and $(\Bf,\Af) \in \{\emptyset, C\}^2$, i.e., $D[i]$ denotes the maximal word in
$D$, $j$ is the number of consecutive occurrences of the corresponding word, and
$\Bf$ and $\Af$ represent the words before and after $D[i]$. These tuples are
lexicographically ordered by the total order relation $>$ on the set of possible
$(i,j,\Bf,\Af)$ tuples defined as follows: let $t=(i,j,\Bf,\Af)$ and
$t'=(i',j',\Bf',\Af')$ two such tuples. Then $t>t'$ iff $w>w'$ or $w=w' \wedge
j>j'$ or $w=w' \wedge j=j' \wedge \Bf={C} \wedge \Bf'={\emptyset}$ or $w=w'
\wedge j=j' \wedge \Bf=\Bf' \wedge \Af={C} \wedge \Af'={\emptyset}$.

With these definition we can prove that algorithm~$\FRAME$ is correct.


  \begin{theorem}\label{thm:main-app}
   Given a dictionary for $\mathcal{H}$, algorithm~$\FRAME$ correctly discovers any $H\in \mathcal{H}$.
  \end{theorem}
  \begin{proof}
  We first prove that the claim is true if $H$ forms a motif sequence (without edge expansion). Subsequently,
we study the case where the motif sequence is expanded by Rule~2, and finally tackle the general composition
case.

\textbf{Discovery of motif sequences:} Due to
Lemma~\ref{lemma:dico-property-app} it holds that for $w$ chosen when Line 1 of
$\mathit{find\_motif\_sequence}()$ is executed for the first time, $S$ is
partitioned into three subsequences $S_1$, $w$ and $S_2$. Subsequently
$\mathit{find\_motif\_sequence}()$ is executed on each of the subsequences
$S'\in\{S_1, S_2\}$ recursively if $C\mapsto S'$, i.e., if the subsequences are
not empty. Thus $\mathit{find\_motif\_sequence}()$ computes a decomposition as
described in Corollary~\ref{cor:recursive-app} recursively. As each of the words
used in the decomposition is a subsequence of $S$ and
$\mathit{find\_motif\_sequence}()$ does not stop until no more words can be
added to any subsequence, it holds that all nodes of $S$ will be discovered
eventually. In other words, $\pi^{-1}(u)$ is defined for all $u\in S$.

As a next step we assume $S' \neq S$ to be the sequence of words obtained by
$\FRAME$ to derive a contradiction. Since $S':=H'$ is the output of
algorithm~$\FRAME$ and is hence embeddable in $H$: $S'\mapsto S$, there
exists a valid embedding mapping $\pi$. Given $u, v\in V(S)$, we denote by
$E^{\pi^{-1}}(S')$ the set of pairs $\{u,v\}$ for which
$\{\pi^{-1}(u),\pi^{-1}(v)\} \in E(S')$. Now assume that $S$ and $S'$ do not
lead to the same resource reservations ``$\pi(S) \neq \pi(S')$''. Hence there
are some inconsistencies between the substrate and the output of
algorithm~$\FRAME$: $\Phi=\{\{u,v\} \in E(S)\backslash E^{\pi^{-1}}(S') \cup
E^{\pi^{-1}}(S') \backslash E(S)\}$. With each of these ``conflict'' edges, one
can associate the corresponding word $W_{u,v}$ (resp.  $W_{u,v}'$) in $S$
(resp. $S'$). If a given conflict edge spans multiple words, we only consider
the words with the highest index as defined by $\FRAME$.  We also define
$i_{u,v}=r(W_{u,v})$ (resp.  $i_{u,v}'=r(W_{u,v}'$)). Since $S'$ and $S$ are by
definition not isomorphic, $i_{u,v}' \neq
i_{u,v}$. 

Let $j = \max_{(u,v) \in \Phi}( i_{u,v})$ be the index of the greatest word
embeddable on the substrate containing an inconsistency, and $j'$ be the index
of the corresponding word detected by $\FRAME$.

($i$) Assume $j>j'$: a lower order motif was erroneously detected. Let $J^+$
(and $J^-$) be the set of dictionary
entries 
that are detected before (after) $D[j]$ (if any) in $S$ by
  $\FRAME$.  Observe that the words in $J^+$ were perfectly detected by
$\FRAME$, otherwise we are in Case~($ii$). We can decompose $S$ as an
alternating sequence of words of $J^+$ and other words using
Corollary~\ref{cor:recursive-app} : $S=T_1 J_1(a_1) T_2 \ldots T_k$ with
$J_i(a_i)\in(J^+)^\star$ and attachment points $a_i$ and $T_i \in
(J^-)^\star$. As the words in $J^+$ are the same in $S'$, we can write $S'=T_1'
J_1 T_2' \ldots T_k'$ (using Corollary~\ref{cor:recursive-app} as well).
  	
  Let $T$ be the sequence among $T_{1}, \ldots, T_k$ that contains our
  misdetected word $D[j]$, and $T'$ the corresponding sequence in $S'$. Observe
  that $T' \mapsto T$ since the words $J_i$ cut the sequences of $S$ and $S'$
  into subsequences $T_i, ~T_i'$ that are embeddable. Observe that $D[j] \mapsto
  T$ since $T$ contains it.  Note that in the execution of
  $\mathit{find\_motif\_sequence}()$ when $D[j']$ was detected the higher
  indexed words had been detected correctly by $\FRAME$ in previous executions
  of this subroutine. Hence, $T_<$ and $T_>$ cannot contain any words
  leading to edges in $\Phi$. Thus $(j',.,.,.)<(j,.,.,.)$ which contradicts Line
  $1$ of $\mathit{find\_motif\_sequence}()$. \stefan{Stefan - Fixme: more obvious there ?
    Maybe: we deduce that $j' \neq \arg\max_x(D[x] \mapsto T)$ since $j =
    \arg\max_x(D[x] \mapsto T)$ and $j'<j$ which contradicts Line $1$ of
    $\mathit{find\_motif\_sequence}()$. }

  ($ii$) Now assume $j'>j$: a higher order motif was erroneously detected. Using
  the same decomposition as step ($i$), we define $J'^+$ as the set of words
  perfectly detected, and therefore decompose $S$ and $S'$ as sequences
  $S=T_1J_1'T_2\ldots J_{k-1}'T_k$ and $S'=T_1'J_1'T_2'\ldots J_{k-1}'T_k'$ with
  $J_i'\in(J'^+)^\star$ and the property that each $T_i' \mapsto T_i $.

  Let $T'$ be the sequence among $T_{1}', \ldots, T_k'$ that contains our
  misdetected word $D[j']$, and $T$ the corresponding sequence in $S$.  Since
  $D[j'] \prec T'$, $D[j'] \mapsto T'$. Thus, since $T' \mapsto T$, we deduce
  $D[j']\mapsto T$ which is a contradiction with $j'$ and
  Corollary~\ref{cor:recursive-app}.

  The same arguments can be applied recursively to show that conflicts in
  $\phi$ of smaller indices cannot exist either.

  \textbf{Expanded motif sequences.} As a next step, we consider graphs that
  have been extended by applying node insertions (Rule 2) to motif sequences, so
  called \emph{expanded} motif sequences: we prove that if $H$ is an
  expanded motif sequence $S$, then algorithm~$\FRAME$ correctly discovers $S$.
  Given an expanded motif sequence $S$, replacing all two degree nodes with an
  edge connecting their neighbors unless a cycle of length three would be
  destroyed, leads to a unique pure motif sequence $T$, $T\mapsto S$.  For the
  corresponding embedding mapping $\pi$ it holds that $V(S)\setminus \pi(T)$ is
  exactly the set $\mathcal{R}$ of removed nodes.  Applying
  $\mathit{find\_motif\_sequence}()$ to an expanded motif sequence discovers
  this pure motif sequence $T$ by using the nodes in $\mathcal{R}$ as relay
  nodes. All nodes in $\mathcal{R}$ are then discovered in
  $\mathit{edge\_expansion}()$ where the reverse operation node insertion is
  carried out as often as possible. It follows that each node in $S$ is either
  discovered in $\mathit{find\_motif\_sequence}()$ if it occurs in a motif
    or in $\mathit{edge\_expansion}()$ otherwise.

\textbf{Combining expanded sequences.} Finally, it remains to combine the expanded sequences.
Clearly, since motifs describe all parts of the graph which are at least 2-connected, the graph
remaining after collapsing motifs cannot contain any cycles: it is a tree. However, on this graph
$\FRAME$ behaves like $\TREE$, but instead of attaching chains, entire sequences are attached to different nodes.
Along the unique sequence paths between two nodes, $\FRAME$ fixes the largest words first, and the claim follows
by the same arguments as used in the proofs for tree and cactus graphs.
  \hfill$\Box$ \end{proof}
\begin{footnotesize}
\vspace{-1.0\baselineskip}
\begin{algorithm}[h!]
\small
\caption{Motif Graph Discovery $\FRAME$}
   \label{alg:motifRec-app}
    \begin{algorithmic}[1]
        \STATE $H':=\{\{v\},\emptyset\}$ ~/*current request graph*/,~~~~~
        $\mathcal{P}: =\{v\}$ ~/*set of unexplored nodes*/
        \WHILE {$\mathcal P\neq \emptyset$}
        \STATE choose $v\in \mathcal P$, $T: =\mathit{find\_motif\_sequence}(v,\emptyset,\emptyset)$
        \STATE \textbf{if } ($T\neq\emptyset$) \textbf{then} $H':=H'vT$, add all nodes of $T$ to $\mathcal{P}$,
        \textbf{for all} {$e\in T$} \textbf{do} \emph{edgeExpansion}($e$)
        \STATE \textbf{else} remove $v$ from $\mathcal{P}$
        \ENDWHILE
    \end{algorithmic}
\vspace{.1cm}
\emph{find\_motif\_sequence}($v,T_<,T_>$)
   \begin{algorithmic}[1]
  \STATE  find maximal
  $i,j,\Bf,\Af$ s.t.~$H'v~(T_<)~\Bf~(D[i])^j~\Af~(T_>)$ \label{line:potenz}
  $\mapsto H$ where $\Bf,\Af \in \{\emptyset,C\}^2$ ~~ /* issue requests */
   \STATE \textbf{if} ($(i,j,\Bf,\Af)=(0,0,C,\emptyset) $) \textbf{then return} $T_<CT_>$
   \STATE \textbf{if } ($\Bf=C$) \textbf{then} $\Bf=\mathit{find\_motif\_sequence}(v,T_<,(D[i])^j~\Af~T_>)$
   \STATE \textbf{if } ($\Af=C$) \textbf{then}  $\Af=\mathit{find\_motif\_sequence}(v,T_<~\Bf~(D[i])^j,T_>)$
   \STATE \textbf{return} $\Bf~(D[i])^j~\Af$
   \end{algorithmic}
   \vspace{.1cm}
\emph{edge\_expansion}($e$)
    \begin{algorithmic}[1]
    \STATE let $u,v$ be the endpoints of edge $e$, remove $e$ from $H'$
    \STATE find maximal $j$ s.t. $H'vC^ju \mapsto H$ ~~ /* issue requests */
    \STATE $H':=H'vC^ju$, add newly discovered nodes to $\mathcal{P}$
    \end{algorithmic}
 \end{algorithm}
\vspace{-2.5\baselineskip}
\end{footnotesize}

\subsection{Request Complexity}

The focus of $\FRAME$ is on generality rather than performance, and indeed,
the resulting request complexities can often be
high. However, as we will see, there are interesting graph classes which can be solved efficiently.

Let us start with a general complexity analysis. The requests issued by $\FRAME$ are constructed in Line 1 of $finding\_motif\_sequence()$ and in Line 2 of $edge\_expansion()$. We will show that the request complexity of the latter is linear in the number of edges of the host graph while the request complexity of $finding\_motif\_sequence()$ depends on the structure of the dictionary. 
Essentially, an efficient implementation of Line 1 of $finding\_motif\_sequence$ in $\FRAME$ can be seen as the depth-first exploration of the dictionary $D$ starting from the chain $C$. More precisely, at a dictionary word $v$ requests are issued to see if one of the outgoing neighbors of $v$ could be embedded at the position of $v$. As soon as one of the replies is positive, we follow the corresponding edge and continue recursively from there, until no outgoing neighbors can be embedded. Thus, the number of requests issued before we reach a vertex $v$ can be determined easily.

Recall that $\FRAME$ tests vertices of a dictionary $D$ according to a fixed
port labeling scheme. For any $v\in V_D$, let $p(C,v)$ be the set of paths
from $C$ to $v$ (each path including $C$ and $v$). In the worst case, discovering $v$ costs $cost(v)=\max_{p\in
  p(C,v)}( \sum_{u \in p}|\Gamma^+(u)|)$. 

\begin{lemma}\label{lem:weight} The request complexity of Line 1 of
  $find\_motif\_sequence(v',T_<,T_>)$ to find the maximal
  $i,j,\Bf,\Af$ such that~$H'v'~(T_<)~\Bf~(D[i])^j~\Af~(T_>) \mapsto H$ where $\Bf,\Af \in \{\emptyset,C\}^2$ and $H'$ is the current request graph is $O(\max_{v\in V_D} cost(v)+j)$.
\end{lemma}
\begin{proof}
  To reach a word $v=D[i]$ in $V_D$ with depth-first traversal there is
  exactly one path between the chain $C$ and $v$.  $\FRAME$ issues a request
  for at most all the outgoing neighbors of the nodes this path. After $v$ has
  been found, the highest $j$ where $H'v~(T_<)~\Bf~(v^j)~\Af~(T_>)\mapsto H$
  has to be determined. To this end, another $j+1$ requests are necessary.
  Thus the maximum of $cost(v)+j$ over all word $v\in V_D$ determines the
  request complexity.
  \hfill$\Box$ \end{proof}

When additional nodes are discovered by a positive reply to an embedding
request, then the request complexity between this and the last previous
positive reply can be amortized among the newly discovered nodes. Let $num\_nodes(v)$ denote the number of nodes in the
motif sequence of the node $v$ in the dictionary.

\begin{theorem}\label{thm:runtime}
 The request complexity of algorithm~$\FRAME$ is at most $O(n\cdot  \Delta+m)$, where $m$ denotes the number of
 edges of the inferred graph $H\in \mathcal{H}$, and $\Delta$ is the maximal
 ratio between the cost of discovering a word $v$  in $D$ and $num\_nodes(v)$, i.e., $ \Delta = \max_{v\in V_D}\{cost(v)/num\_nodes(v)\}$.
\end{theorem}

\begin{proof}
Each time Line~\ref{line:potenz} of $find\_motif\_sequence()$ is called, either
at least one new node is found or no other node can be embedded between
the current sequences (one request is necessary for the latter result). If
one or more new nodes are discovered, the request complexity can be amortized
by the number of nodes found: If $v$ is the maximal word found in Line 1 of
$find\_motif\_sequence()$ then it is responsible for at most $cost(v)$ requests due to
Lemma~\ref{lem:weight}. If it occurs more than once at this position, only
one additional request is necessary to discover even more nodes (plus one
superfluous request if no more occurrences of $v$ can be embedded there).
Amortizing the request number over the number of discovered nodes results in
$\Delta$ requests.
All other requests are due to $\mathit{edge\_expansion}(e)$  where additional
nodes are placed along edges. Clearly, these costs can be amortized by the
number of edges in $H$: for each edge $e\in E(H)$, at most two embedding
requests are performed (including
a ``superfluous'' request which is needed for termination when no additional
nodes can be added).
\hfill$\Box$ \end{proof}

\subsection{Examples}

Let us consider concrete examples to provide some intuition for Theorem~\ref{thm:main-app} and Theorem~\ref{thm:runtime}.
\stefan{If we ignore the attachment points, then executing DICT for the graph in Figure 2\emph{b}, leads to the following request sequence.
Chains, cycles, diamonds, the complete bipartite graph over two times three nodes and the complete graph over five nodes are denoted by $C$, $Y$, $D$ and $B$ and $K$ respectively.\\
$\{$ 1. $C$, 2. $Y$, 3. $D$, 4. $B$, 5. $K$ (negative reply), 6. $BB$ (negative reply), 7. $C K C$ (recursion in Line~3), 8. $Y K C$ (negative reply), 9. $Y C  K  C$, 10. $D C K C$, 11. $D^2 C K C$  (negative reply), 12. $C D C K C$, $\ldots\}$\\
}
The execution of $\FRAME$ for the graph in Figure~\ref{fig:dictionary}.b), is illustrated in Figure~\ref{fig:motiftree}.
\begin{figure}[h]
	\centering
		\includegraphics[width=0.40\columnwidth]{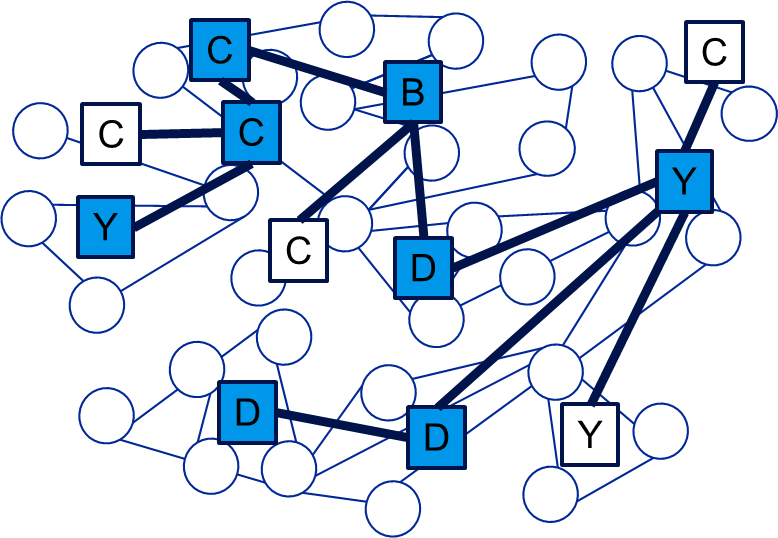}
\vspace{-0.5\baselineskip}	\caption{Motif sequence tree of the graph in Figure~\ref{fig:dictionary} \emph{b}). The squares and the edges between them depict the motif composition, the shaded squares belong to the motif sequence  $YC^2BDYD^2$ discovered in the first execution of $find\_motif\_sequence()$ (chains, cycles, diamonds, and the complete bipartite graph over two times three nodes are denoted by $C$, $Y$, $D$ and $B$ respectively). Subsequently, the found edges are expanded before calling $find\_motif\_sequence()$ another four times to find $Y$ and three times $C$.}
	\label{fig:motiftree}
\vspace{-1.0\baselineskip}
\end{figure}

A fundamental graph class are \emph{trees}. Since, the tree
does not contain any 2-connected structures, it can be described by a
single motif: the chain $C$. Indeed, if $\FRAME$ is executed with a dictionary
consisting in the singleton motif set $\{C\}$, it is equivalent to a recursive version of $\TREE$ from~\cite{minicom13} and seeks to compute maximal paths. For the cactus graph, we have two motifs, the request complexity is the same as for the algorithm described in~\cite{minicom13}.

\begin{corollary}
Trees can be described by one motif (the chain $C$), and cactus graphs by two motifs (the chain $C$ and the cycle $Y$). The request complexity of $\FRAME$ on trees and cactus graphs is $O(n)$.
\end{corollary}
\begin{proof}
We present the arguments for cactus graphs only, as trees constitute a subset of the cactus family. The absence of diamond graph minors implies that a cactus graph does not
contain two closed faces which share a link. Thus, there can exist at most two
different (not even disjoint) paths between any node pair, and the
corresponding motif subgraph forms a \emph{cycle} $Y$ (or a triangle). Since the cycle has only one attachment point pair, $\Delta$ of $D$ is constant. Consequently, a linear request complexity follows directly from Theorem~\ref{thm:runtime} due to the planarity of cactus graphs  (i.e., $m\in O(n)$).
\hfill$\Box$ \end{proof}


An example where the dictionary is
efficient although the connectivity of the topology can be high are \emph{block graphs}. A block graph is an undirected graph in which
every bi-connected component (a \emph{block}) is a clique. A \emph{generalized
block graph} is a block graph where the edges of the cliques can contain
additional nodes. In other words, in the terminology of our framework, the
motifs of generalized block graphs are \emph{cliques}. For instance, cactus
graphs are generalized block graphs where the maximal clique size is three.

\begin{corollary}
Generalized block graphs can be described by the motif set of cliques. The request complexity of $\FRAME$ on generalized block graphs is $O(m)$, where $m$ denotes the number of edges in the host graph.
\end{corollary}
\begin{proof}
The framework dictionary for generalized block graphs consists of the
set of cliques, as a clique with $k$ nodes cannot be embedded on sequences of
cliques with less than $k$ nodes. As there are three attachment point pairs for each complete graph with four or more nodes,
$\FRAME$ can
be applied using a dictionary that contains three entries for each motif with more than three nodes ($num\_nodes()>3$). Thus, the $i^{th}$ dictionary entry has $\lfloor i/3 \rfloor+3$ nodes for $i>1$ and $cost(D[i]) < 3(i+2)$ and $\Delta$ of $D$ is hence in $O(1)$. Consequently the complexity for generalized block graphs is $O(m)$ due to Theorem~\ref{thm:runtime}.
\hfill$\Box$ \end{proof}

On the other hand, Theorem~\ref{thm:runtime} also states that highly connected graphs may require $\Omega(n^2)$ requests, even if the dictionary is small. In the next section, we will study whether this happens in ``real world graphs''.

\begin{figure}[b!]
	\centering
a)\includegraphics[width=0.42\textwidth]{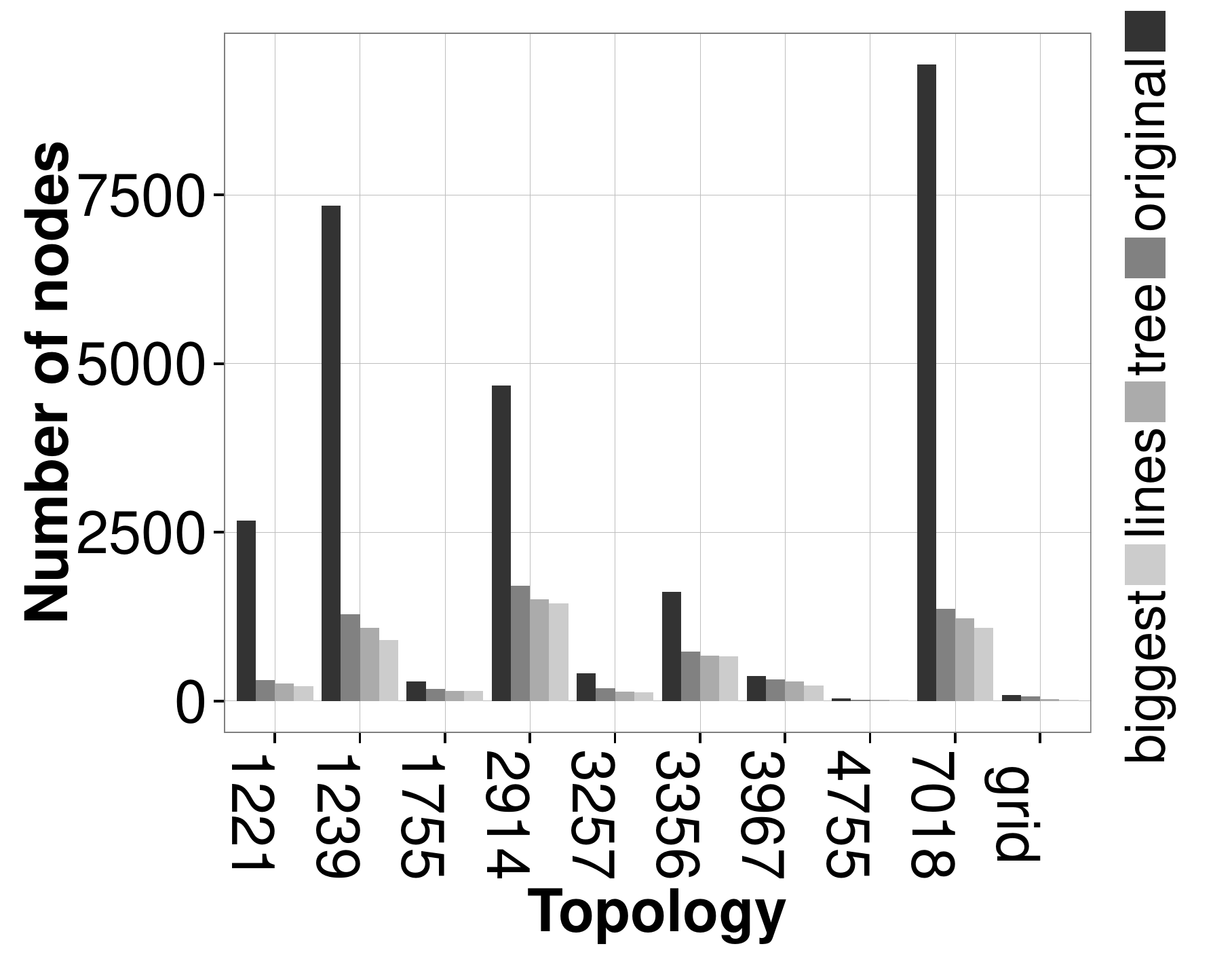}
b)\includegraphics[width=0.42 \textwidth]{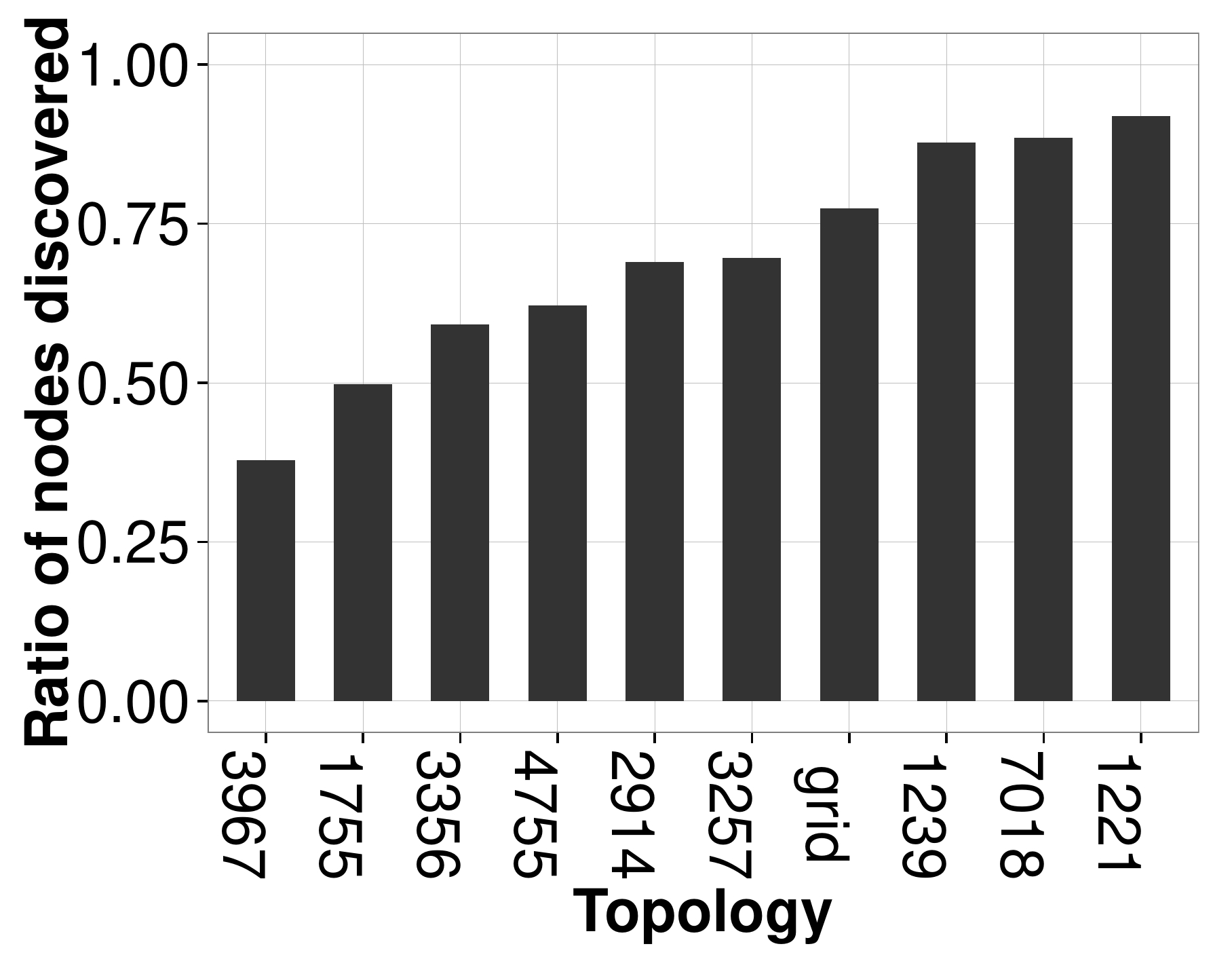}
c)\includegraphics[width=0.4\textwidth]{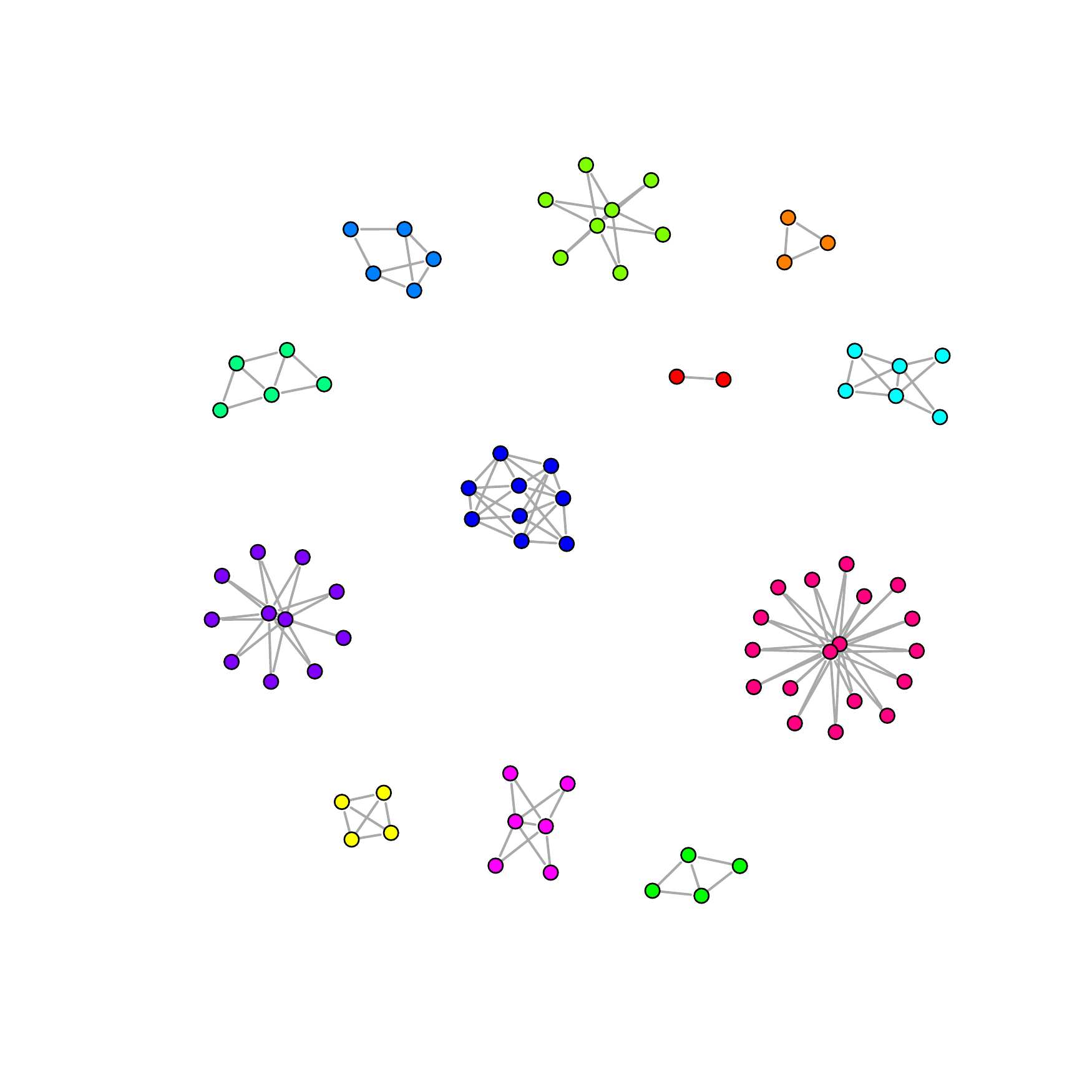}
d)\includegraphics[width=0.4 \textwidth]{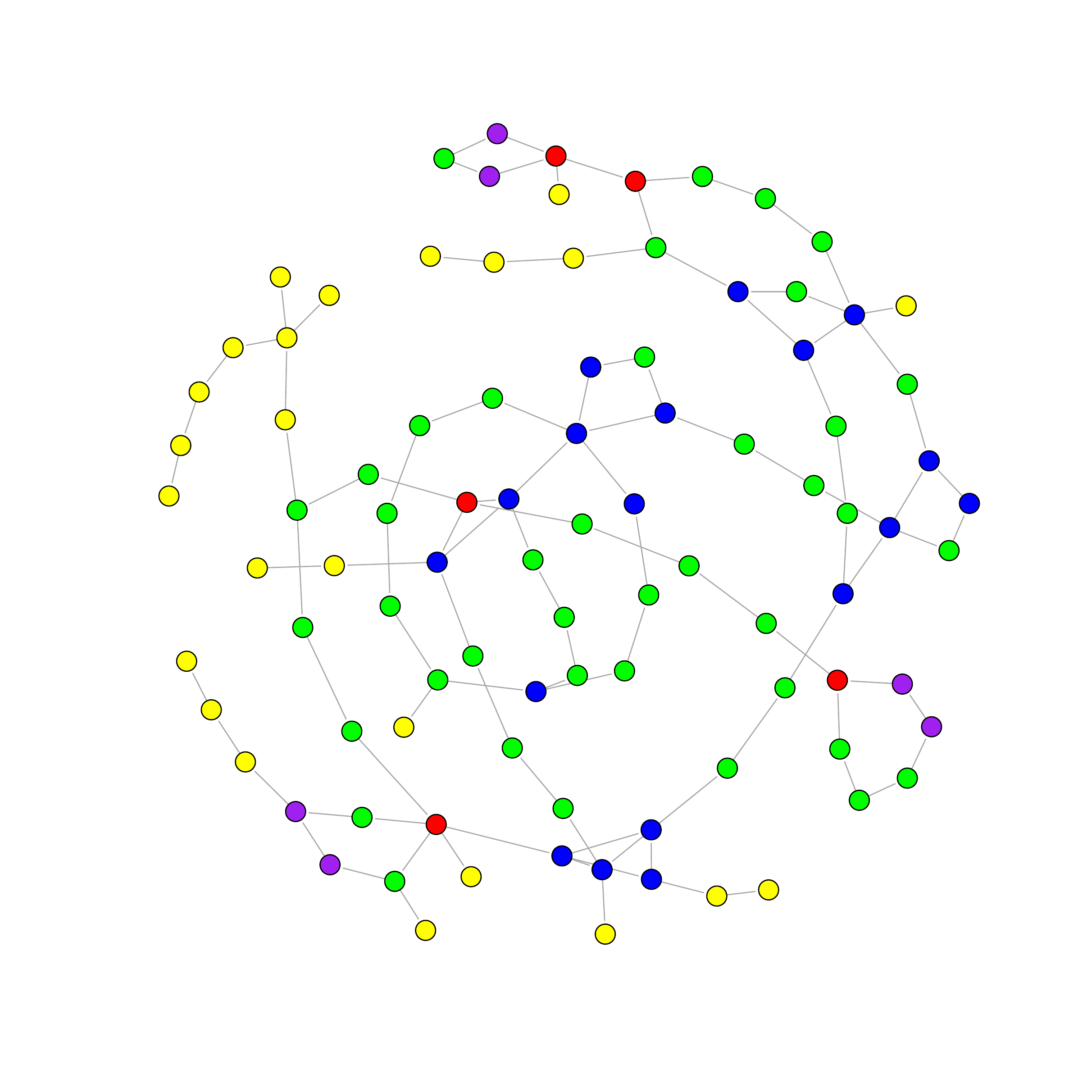}
\vspace{-0.5\baselineskip}	\caption{Results of $\FRAME$ when run on different Internet and power grid topologies. a)
Number of nodes in different autonomous systems (AS). We
computed the set of motifs of these graphs as described in
Definition~\ref{def:motif-app} and counted the number of nodes that: (i) belong
to a tree structure at the fringe of the network, (ii) have degree 2 and belong
to two-connected motifs, and finally (iii) are part of the largest
motif.  b) The fraction of nodes that can be discovered with 12-motif 
dictionary represented in Figure c). d) An example network where tree nodes 
are colored yellow, line-nodes are green, attachment point nodes are red and 
the remaining nodes blue.}
	\label{fig:exp}
\vspace{-1.0\baselineskip}
\end{figure}

\section{Experiments}\label{sec:experiments}
To complement our theoretical results and to validate our framework on
realistic graphs, we dissected the \emph{ISP topologies} provided by the
Rocketfuel mapping engine\footnote{See
  \url{http://www.cs.washington.edu/research/networking/rocketfuel/}.}.
In addition, we also dissected the topology of a European electricity distribution grid
  (\texttt{grid} on the legends). Figure~\ref{fig:exp} a) provides some
  statistics about the aforementioned topologies. Since $\FRAME$ discovers 
  both tree and degree 2 nodes in linear time, this figure shows that most of 
  each topology can be discovered quickly.
The inspected topologies are composed of a large bi-connected component (the
  largest motif), and some other small and simple motifs. Figure~\ref{fig:exp}~b) represents the fraction of each topology that can be discovered by $\FRAME$
  using only a 12-motifs dictionary (see~Figure~\ref{fig:exp}~c)). Interestingly, this small dictionary is efficient on
  $10$ different topologies, and contains motifs that are mostly
  symmetrical. This might stem from the man-engineered origin of the targeted
  topologies. Finally, Figure~\ref{fig:exp} d) provides an example of such
  a topology.

\begin{tiny}
\bibliographystyle{abbrv}

\end{tiny}
\newpage

\appendix
\section{Appendix}

\textbf{Lemma~\ref{lemma:poset}.}\emph{
The embedding relation $\mapsto$ applied to any family $\mathcal{G}$ of undirected graphs (short: $(\mathcal{G},\mapsto)$), forms a partially ordered set (a \emph{poset}).
}
\begin{proof}
A poset structure $(S,\preceq)$ over a set $S$ requires that $\preceq$ is a (\emph{reflexive}, \emph{transitive}, and \emph{antisymmetric}) order which may or
 may not be partial. To show that $(\mathcal{G},\mapsto)$, the embedding order defined over a given set of graphs $\mathcal{G}$, is a poset, we  examine the
three properties in turn.

\emph{Reflexive} $G\in \mathcal{G} \mapsto G \in \mathcal{G}$: By using the identity mapping $\pi: G=(V,E) \rightarrow G=(V,E)$ which embeds each node and link
to itself, the claim is proved.

\emph{Transitive} $A\in \mathcal{G} \mapsto B \in \mathcal{G}$ and $B\in \mathcal{G} \mapsto C \in \mathcal{G}$
implies $A\in \mathcal{G} \mapsto C \in \mathcal{G}$: Let $\pi_1$ denote the embedding function for $A\in \mathcal{G}
\mapsto B \in \mathcal{G}$ and let $\pi_2$ denote the embedding function for $B\in \mathcal{G} \mapsto C \in \mathcal{G}$,
 which must exist by our assumptions. We will show that then also a valid embedding function $\pi$ exists to map $A$ to $C$. Regarding the
 node mapping, we define $\pi_V$ as $\pi_V:=\pi_{1V} \circ \pi_{2V}$, i.e., the result of first mapping the nodes according to $\pi_{1V}$ and
 subsequently according to $\pi_{2V}$. We first show that $\pi_V$ is a valid mapping from $A$ to $C$ as well. First, $\forall v_A \in V_A$, $\pi(v_A)$
 maps $v_A$ to a single node in $V_C$, fulfilling the first condition of the embedding (see Definition~\ref{def:embedding}). Ignoring relay capacities
 (which is studied together with the conditions on the links below),
Condition~($ii$) of Definition~\ref{def:embedding} is also fulfilled since the mapping $\pi_{1V}$ ensures that no node in $V_B$ exceeds its capacity, and
can hence safely be mapped to $V_C$. Let us now turn our attention to the links. We use the following mapping $\pi_E$ for the edges. Note that $\pi_{1E}$
maps a single link $e$ to an entire (but possibly empty) path in $B$ and $\pi_{2E}$ then maps the corresponding links $e'$ in $B$ to a walk in $C$. We can
transform any of these walks into paths by removing cycles; this can only lower the resource costs. Since $\pi_{1E}$ maps to a subset of $E_B$ only and
since $\pi_{2E}$ can embed all edges of $B$, all link capacities are respected up to relay costs. However, note also that by the mapping $\pi_1$ and for
relay costs $\epsilon>0$, each node $v_B\in V_B$ can either not be used at all, be fully used as a single endpoint of a link $e_A\in E_A$, or serve as a
relay for one or more links. Since both end-nodes and relay nodes are mapped to separate nodes in $C$, capacities are respected as well. Conditions~($iii$)
 and ($iv$) hold as well.

\emph{Antisymmetric} $A\in \mathcal{G} \mapsto B \in \mathcal{G}$ and $B\in \mathcal{G} \mapsto A \in \mathcal{G}$ implies $A=B$, i.e., $A$ and
$B$ are isomorphic and have the same weights: First observe that if the two networks differ in size, i.e., $|V_A|\neq |V_B|$ or $|E_A|\neq |E_B|$,
then they cannot be embedded to each other: W.l.o.g., assume $|V_A|>|V_B|$, then since nodes of $V_A$ of cannot be split into multiple nodes of $V_B$
(cf Definition~\ref{def:embedding}), there exists a node $v_A\in V_A$ to which no node from $V_B$ is mapped. This however implies that node $\pi_1(v_A)\in V_B$
must have available capacities to host also $v_A$, contradicting our assumption that nodes cannot be split in the embedding. Similarly, if $|E_A|\neq |E_B|$, we
can obtain a contradiction with the single path argument. Thus, not only the total number of nodes and links in $A$ and $B$ must be equivalent but also the total
amount of node and link resources. So consider a valid embedding $\pi_1$ for $A\in \mathcal{G} \mapsto B \in \mathcal{G}$ and a valid embedding $\pi_2$ for $B\in
\mathcal{G} \mapsto A \in \mathcal{G}$, and assume $|V_A|=|V_B|$ and $|E_A|=|E_B|$. It holds that $\pi_1$ and $\pi_2$ define an isomorphism between $A$ and $B$:
Clearly, since $|V_A|=|V_B|$, $\pi_1$ and $\pi_2$ define a permutation on the vertices. W.l.o.g., consider any link $\{v_A,v_A'\}\in E_A$. Then, also $\{\pi_{1}(v_A),
\pi_{1}(v_A')\}\in E_B$: $|\{\pi_{1}(v_A),\pi_{1}(v_1')\}|=0$ would violate the node capacity constraints in $B$, and $|\{\pi_{1}(v_A),\pi_{1}(v_A')\}|>1$ requires
$|E_B|>|E_A|$.
\hfill$\Box$ \end{proof}

\textbf{Lemma~\ref{lem:dicoExistence}.} \emph{
There exists a dictionary $D=(V_D,E_D)$ that covers all member graphs $H$ of a motif graph family $\mathcal{H}$ with $n$ vertices.}
\begin{proof}
We present a procedure to construct such a dictionary $D$. Let $\mathcal{M}_n$
be the set of all motifs with $n$ nodes of the graph family $\mathcal{H}$. For
each motif $m\in\mathcal{M}_n$ with $x$ possible attachment point pairs (up to
isomorphisms), we add $x$ dictionary words to $V_D$, one for each attachment
point pair. The resulting set is denoted by $V_M$. For each sequence of $V_M^\star$ with at most $n$ nodes, we add another word to $V_D$ (with the un-used attachment points of the first and the last subword). There is an edge $e\in E_D$ if the transitive reduction of the embedding relation with context includes an edge between two words.
 We now prove that $D$ is a dictionary, i.e., it is robust to composition. Let $i \in  V_D$. Observe that $R_i$ contains all compositions of words with at most $n$ nodes in which $i$ can be embedded. Consequently, no matter which sequences are in $\overline{R}_i^\star$ it holds that $v_i$ cannot be embedded in a sequences in $Q_i$ the robustness condition is satisfied. Since $H$ has $n$ vertices, and since $D$ contains all possible motifs of at most $n$ vertices, $D$ covers $H$.
\hfill$\Box$ \end{proof}

Note that the proof of Lemma~\ref{lem:dicoExistence} only addresses the composition
robustness for sequences of up to $n$ nodes. However, it is clear that $|V(G)| >
|V(H)| \Rightarrow G \not \mapsto H$, and therefore no ``mismatch'' can happen to
happen. Finite dictionaries and with this adapted
composition can also be applied in the lemmata proved above, there is only a
small notational change necessary in the proof of
Lemma~\ref{lemma:dico-property-app}. (Note that it is always possible to determine
the number of nodes $n$ by binary search using $O(\log{n})$ requests.)

\end{document}